\documentclass{article}

\usepackage[affil-it]{authblk}
\usepackage[usenames,dvipsnames]{xcolor}
\usepackage{amsfonts}
\usepackage{amsmath,amsthm,amssymb,dsfont}
\usepackage{enumerate}
\usepackage[english]{babel}
\usepackage{graphicx}	
\usepackage[caption=false]{subfig}
\usepackage[margin=3cm]{geometry}
\usepackage{url}
\usepackage{todonotes}
\usepackage{bbm}

\usepackage{tikz}
\usetikzlibrary{chains}
\usetikzlibrary{fit}
\usepackage{pgflibraryarrows}		
\usepackage{pgflibrarysnakes}		

\usepackage{makecell}

\usepackage{epsfig}
\usetikzlibrary{shapes.symbols,patterns} 
\usepackage{pgfplots}

\usepackage{hyperref}
\hypersetup{colorlinks=true,citecolor=blue,linkcolor=blue,filecolor=blue,urlcolor=blue,breaklinks=true}

\usepackage{nicefrac}

\theoremstyle{plain}
\newtheorem{theorem}{Theorem}[section]
\newtheorem{lemma}[theorem]{Lemma}

\theoremstyle{definition}

\newtheorem{remark}[theorem]{Remark}

\newcommand*{\PP}{\mathbb{P}}

\newcommand*{\ee}{\mathrm{e}}

\newcommand*{\cE}{\mathcal{E}}

\newcommand*{\cH}{\mathcal{H}}

\newcommand*{\cI}{\mathcal{I}}

\newcommand*{\cM}{\mathcal{M}}
\newcommand*{\cP}{\mathcal{P}}

\newcommand*{\cR}{\mathcal{R}}
\newcommand*{\cS}{\mathcal{S}}
\newcommand*{\cT}{\mathcal{T}}
\newcommand*{\cU}{\mathcal{U}}

\newcommand*{\cX}{\mathcal{X}}

\newcommand*{\N}{\mathbb{N}}

\newcommand*{\R}{\mathbb{R}}

\newcommand*{\St}{\mathrm{S}}

\newcommand*{\eps}{\varepsilon}

\newcommand*{\id}{\mathrm{id}}

\newcommand*{\supp}{\mathrm{supp}}
\newcommand*{\tr}{\mathrm{tr}}
\newcommand*{\ket}[1]{| #1 \rangle}

\newcommand{\proj}[1]{|#1\rangle\!\langle #1|}

\newcommand*{\Pos}{\mathrm{P}}

\newcommand*{\QMC}{\mathrm{MC}}

\newcommand*{\MD}{D_{\mathbb{M}}}

\newcommand*{\ci}{\mathrm{i}} 
\newcommand*{\di}{\mathrm{d}} 

\newcommand{\norm}[1]{\left\lVert#1\right\rVert}

\allowdisplaybreaks    

\begin{document}

\title{\LARGE Necessary criterion for approximate recoverability}

\author{David Sutter}
\author{Renato Renner}

\affil{\small{Institute for Theoretical Physics, ETH Zurich, Switzerland}}
\date{}

\maketitle


\begin{abstract}
A tripartite state $\rho_{ABC}$ forms a Markov chain if there exists a recovery map $\cR_{B \to BC}$ acting only on the $B$-part that perfectly reconstructs $\rho_{ABC}$ from $\rho_{AB}$. To achieve an approximate reconstruction, it suffices that the conditional mutual information $I(A:C|B)_{\rho}$ is small, as shown recently. Here we ask what conditions are necessary for approximate state reconstruction. This is answered by a lower bound on the relative entropy between $\rho_{ABC}$ and the recovered state $\cR_{B\to BC}(\rho_{AB})$. The bound consists of the conditional mutual information and an entropic correction term that quantifies the disturbance of the $B$-part by the recovery map.

\end{abstract}
\
\section{Introduction}
A recovery map is a trace-preserving completely positive map that reconstructs parts of a composite system. More precisely, for a tripartite state $\rho_{ABC}$ on $A\otimes B \otimes C$ we can consider a recovery map $\cR_{B \to BC}$ from $B$ to $B\otimes C$ that reconstructs the $C$-part from the $B$-part only. If such a reconstruction is perfectly possible, i.e., if 
\begin{align} \label{eq_Markov}
\rho_{ABC} =  \cR_{B \to BC}(\rho_{AB}) 
\end{align}
we call $\rho_{ABC}$ a \emph{(quantum) Markov chain} in order $A \leftrightarrow B \leftrightarrow C$.\footnote{We usually omit the identity map and the identity operator in our notation when its use is clear from the context. For example, we write $\cR_{B \to BC}(\rho_{AB})$ instead of $(\cI_A \otimes \cR_{B \to BC})(\rho_{AB})$ and $\rho_B \rho_{AB} \rho_{B}$ instead of $(\id_A \otimes \rho_B) \, \rho_{AB} \, (\id_A \otimes \rho_B)$. We will drop the order of the Markov chain if it is $A \leftrightarrow B \leftrightarrow C$.} 

The structure of Markov chains is well understood. A state $\rho_{ABC}$ is a Markov chain if and only if there exists a decomposition of the $B$ system as $B = \oplus_{j} (b_j^L \otimes b_j^R)$ such that 
\begin{align} \label{eq_MarkovDec}
\rho_{ABC} = \bigoplus_j P(j) \, \rho_{A b_j^L} \otimes \rho_{b_j^R C} \, ,
\end{align}
with states $\rho_{A b_j^L}$ on $A \otimes b_j^L$, $ \rho_{b_j^R C}$ on $b_j^R \otimes C$, and a probability distribution $P$~\cite{HJPW04}.
A measure that is useful to describe Markov chains is the \emph{conditional mutual information} that is given by 
\begin{align}
  I(A : C | B) = \tr\, \rho_{A B C} \bigl( \log \rho_{A B C}  + \log \rho_B  - \log \rho_{A B}  - \log \rho_{B C}  \bigr)  \ ,
\end{align}
whenever the trace is defined, i.e., whenever the operator $\rho_{A B C} ( \log \rho_{A B C}  + \log \rho_B  - \log \rho_{A B}  - \log \rho_{B C})$ is trace class. One often restricts  to the case where the \emph{conditional von Neumann entropy} $H(A | B) = - D(\rho_{A B} \| {\id_A \otimes \rho_B})$ is finite, where $D(\rho\|\sigma):=\tr \rho (\log \rho - \log \sigma)$ denotes the relative entropy between $\rho$ and $\sigma$. Indeed, in this case, the data processing inequality~\cite{lindblad75,uhlmann77} implies that $H(A | B C) = -D(\rho_{A B C} \| {\id_A \otimes \rho_{B C}})$ is also finite, and hence the operators $\rho_{A B C} (\log \rho_{A B} - \log \rho_B)$ and $\rho_{A B C}(\log \rho_{A B C}  - \log \rho_{B C})$ are both trace class, implying that their difference is trace class, too. 
 We further note that for finite-dimensional Hilbert spaces the conditional mutual information may be written as $I({A:C|B})_{\rho}:=H(AB)_{\rho} + H(BC)_{\rho} - H(B)_{\rho} - H(ABC)_{\rho}$ where $H(A)_{\rho} := - \tr \, \rho_A \log \rho_A$ is the von Neumann entropy of the marginal state on $A$.

It has been shown that a state $\rho_{ABC}$ is a Markov chain if and only if its conditional mutual information $I(A:C|B)_{\rho}$ vanishes~\cite{Pet86,Pet03}. Furthermore, the \emph{Petz recovery map} (also known as \emph{transpose map})
\begin{align} \label{eq_PetzRecoveryMap}
\cT_{B\to BC}\, : \, X_B \mapsto \rho_{BC}^{\frac{1}{2}} ( \rho_B^{-\frac{1}{2}} X_B \rho_B^{-\frac{1}{2}} \otimes \id_C ) \rho_{BC}^{\frac{1}{2}}
\end{align}
recovers such states perfectly, i.e.,~\eqref{eq_Markov} holds with $\cR_{B\to BC} = \cT_{B\to BC}$.

Tripartite states $\rho_{ABC}$ that have a small conditional mutual information are called \emph{approximate Markov chains}. The justification for this terminology is a recent result~\cite{FR14} proving that for any state $\rho_{ABC}$ there exists a recovery map $\cR_{B \to BC}$ such that
\begin{align} \label{eq_FR}
I(A:C|B)_{\rho} \geq - \log F\big(\rho_{ABC},\cR_{B\to BC}(\rho_{AB})\big) \, ,
\end{align}
where $F(\tau,\omega):=\norm{\sqrt{\tau} \sqrt{\omega}}^2_1$ denotes the \emph{fidelity} between $\tau$ and $\omega$.\footnote{Recall that for any two states $\tau$ and $\omega$ we have $F(\tau,\omega) \in [0,1]$ and that $F(\tau,\omega)=1$ if and only if $\tau = \omega$ (see, e.g.,~\cite{nielsenChuang_book}).} Inequality~\eqref{eq_FR} shows that the Markov property~\eqref{eq_Markov} approximately holds whenever the conditional mutual information is small. 
However, there exist tripartite states with a small conditional mutual information whose distance to any Markov chain is nevertheless large~\cite{CSW12,ILW08}. As a consequence, approximate quantum Markov chains are not necessarily close to quantum Markov chains. We refer to Appendix~\ref{app_closeMarkov} for a more detailed explanation of this phenomenon.

Inequality~\eqref{eq_FR} has been refined in a series of works~\cite{BHOS14,TB15,SFR15,wilde15,STH15,JRSWW15,SBT16}. More precisely,  the initial bound from~\cite{FR14} has been strengthened by replacing the right-hand side of~\eqref{eq_FR} by the measured relative entropy between the original and the recovered state (see~\eqref{eq_measrelentropy} below for a definition). This result came with a novel  proof based on the notion of  quantum state redistribution~\cite{BHOS14}. The proof has later been simplified by utilizing tools from semidefinite programming~\cite{TB15}. In~\cite{SFR15} it was shown that there exists a universal recovery map, i.e., one that does not depend on the $A$ system, that satisfies~\eqref{eq_FR}. Another major step was the discovery that~\eqref{eq_FR}, as well as generalisations thereof, can be obtained by complex interpolation theory~\cite{wilde15}, providing further insight into the structure of the recovery map. In~\cite{STH15} an intuitive proof of~\eqref{eq_FR} based on the \emph{spectral pinching method} was presented. In~\cite{JRSWW15} it was shown that there exists an explicit recovery map (of the form~\eqref{eq_explicitRecoveryMap}) that satisfies~\eqref{eq_FR}. 
The most recent result~\cite[Theorem~4.1]{SBT16} shows that for any state $\rho_{ABC}$ we have
\begin{align} \label{eq_SBT}
I(A:C|B)_{\rho} \geq \MD\big(\rho_{ABC} \| \cP_{B\to BC}(\rho_{AB}) \big) \, ,
\end{align}
 for the explicit recovery map
 \begin{align} \label{eq_explicitRecoveryMap}
 \cP_{B \to BC}(\cdot):=\int_{-\infty}^{\infty} \beta_0(\di t) \,  \cP^{[t]}_{B \to BC}(\cdot) \quad \text{with} \quad \cP^{[t]}_{B \to BC}(\cdot)= \rho_{BC}^{\frac{1+\ci t}{2}} ( \rho_B^{-\frac{1+\ci t}{2}} X_B \rho_B^{-\frac{1-\ci t}{2}} \otimes \id_C ) \rho_{BC}^{\frac{1-\ci t}{2}} 
 \end{align}
 and the probability measure
 \begin{align} \label{eq_beta0}
\beta_0(\di t) := \frac{\pi}{2} \left( \cosh(\pi t) + 1\right)^{-1} \di t 
\end{align}
on $\R$.
$\MD$ denotes the \emph{measured relative entropy}, which is defined as 
\begin{align} \label{eq_measrelentropy}
\MD(\rho \| \sigma) := \sup \limits_{M \in \cM}  D( M(\rho) \| M(\sigma) ) \ ,
\end{align}
where $\cM$ is the set of all quantum-classical channels $M(\omega)= \sum_x (\tr M_x \omega) \proj{x}$ with $\{M_x\}$ a positive operator valued measure (POVM) and $\{\ket{x}\}$ an orthonormal basis.
A simple property of the measured relative entropy ensures that $\MD(\tau\| \omega) \geq - \log F(\tau,\omega)$ for all states $\tau, \omega$~\cite{BHOS14}, which shows that~\eqref{eq_SBT} implies~\eqref{eq_FR}. We further note that the recovery map $\cP_{B \to BC}$ given in~\eqref{eq_explicitRecoveryMap} is \emph{universal} in the sense that it only depends on $\rho_{BC}$ and it satisfies $ \cP_{B \to BC}(\rho_B)=\rho_{BC}$.
The interested reader can find additional information about the concepts and achievements around~\eqref{eq_FR} in~\cite{sutter_phd}. 

Inequality~\eqref{eq_FR} shows that there always exists a recovery map whose recovery quality (measured in terms of the logarithm of the fidelity) is of the order of the conditional mutual information. This shows that a small conditional mutual information is a sufficient condition for a state to be approximately recoverable.
In other words,~\eqref{eq_FR} gives an entropic characterization for the set of tripartite states that can be approximately recovered. 

In this work, we are interested in an opposite statement. This corresponds to an inequality that bounds the distance between $\rho_{ABC}$ and any reconstructed state $\cR_{B \to BC}(\rho_{AB})$ from below with an entropic functional of $\rho_{ABC}$ and the recovery map $\cR_{B \to BC}$ that involves the conditional mutual information. Such an inequality is the converse to~\eqref{eq_FR}, and gives a necessary condition for approximate recoverability.

\subsection{Main result}
For any trace-preserving completely positive map $\cE$ on a system $S$ we denote by $\mathrm{Inv}(\cE)$ the set of density operators $\tau$ on $S$ which are left invariant under the action of $\cE$, i.e., 
\begin{align}
  \mathrm{Inv}(\cE) := \{\tau: \, \cE(\tau) = \tau\} \ .
\end{align}
We may now quantify the deviation of any state $\rho$ from the set $\mathrm{Inv}(\cE)$ by
\begin{align} \label{eq_Lambda}
  \Lambda_{\max}(\rho \| \cE) := \inf_{\tau \in \mathrm{Inv}(\cE)}  D_{\max}(\rho \| \tau) \ ,
\end{align}
where $D_{\max}(\omega \| \sigma):=\inf\{\lambda \in \R : \omega \leq 2^{\lambda} \sigma \}$ denotes the  the \emph{max-relative entropy}.
The $\Lambda_{\max}$-quantity has the property that it is zero if and only if $\cE$ leaves $\rho$ invariant\footnote{Note that the max-relative entropy has a definiteness property which ensures that for a sequence $(\omega_k)_{k\in \N}$ of states such that $\lim_{k \to \infty} D_{\max}(\tau\|\omega_k) = 0$ we have $\lim_{k \to \infty} \omega_k = \tau$. This follows from the fact that $-\log F(\tau,\omega) \leq D_{\max}(\tau\| \omega)$~\cite{araki_82,berta16} and the definiteness property of the fidelity~\cite{Uhl76,Alberti1983}, i.e., $\lim_{k\to \infty}F(\tau,\omega_k)=1$ implies $\lim_{k\to \infty}\omega_k=\tau$.}, i.e., 
\begin{align}
  \Lambda_{\max}(\rho \| \cE)  = 0 \quad \iff \quad \cE(\rho) = \rho \ .
\end{align}

\paragraph{Main result.}
We prove that for any state $\rho_{ABC}$ on $A\otimes B \otimes C$ and any recovery map $\cR_{B \to BC}$ from the $B$ system to the $B\otimes C$ system we have 
\begin{align} \label{eq_MainResIntro}
D\big(\rho_{ABC} \| \cR_{B\to BC}(\rho_{AB}) \big) + \Lambda_{\max}(\rho_{AB}\| \cR_{B\to B}) \geq I(A:C|B)_{\rho}  \, ,
\end{align}
where $D(\tau\|\sigma):= \tr\, \tau \log \tau - \tr\, \tau \log \sigma$ if $\supp(\tau) \subseteq \supp(\sigma)$ and $+\infty$ otherwise denotes the \emph{relative entropy}, and $\cR_{B\to B}:= \tr_C \circ \cR_{B\to BC}$ is the action of the recovery map $\cR_{B\to BC}$ on $B$. 
We refer to Theorem~\ref{thm_main} for a more precise statement.

\paragraph{Cases where the $\Lambda_{\max}$-term vanishes.}
To interpret the term $\Lambda_{\max}$ in~\eqref{eq_MainResIntro}, note that the recovery map $\cR_{B \to B C}$ generally not only reads the content of system $B$ in order to generate $C$, but  also disturbs it. $\Lambda_{\max}$ quantifies the amount of this disturbance of $B$, taking system $A$ as a reference. In particular, $\Lambda_{\max}(\rho_{AB}\| \cR_{B\to B})=0$ if $\cR_{B \to B C}$ is ``read only'' on $B$, i.e., if $\rho_{AB}= \cR_{B\to B}(\rho_{AB})$. Inequality~\eqref{eq_MainResIntro} then simplifies to
\begin{align}
D\big(\rho_{ABC} \| \cR_{B\to BC}(\rho_{AB}) \big) \geq I(A:C|B)_{\rho}   \, .
\end{align}
We further note that in case $\cR_{B\to BC}$ is a recovery map that is ``read only'' on $B$ its output state $\sigma_{ABC}:=\cR_{B\to BC}(\rho_{AB})$ is a Markov chain since
\begin{align} \label{eq_whyMarkov}
H(A|B)_{\rho} \leq H(A|BC)_{\sigma} \leq H(A|B)_{\sigma}=H(A|B)_{\rho} \, ,
\end{align}
where the two inequality steps follow from the data-processing inequality~\cite{LieRus73_1,LieRus73} applied for $\cR_{B\to BC}$ and $\tr_C$, respectively and hence $I(A:C|B)_{\sigma}= H(A|B)_{\sigma} - H(A|BC)_{\sigma} = 0$.

\subsection{Tightness of the main result}
We next discuss several aspects concerning the tightness of~\eqref{eq_MainResIntro}. This will also give a better understanding about the role of the $\Lambda_{\max}$-term.
We first note that by combining~\eqref{eq_SBT} with~\eqref{eq_MainResIntro} we obtain
\begin{align}
\MD\big(\rho_{ABC} \| \cP_{B\to BC}(\rho_{AB}) \big) & \leq I(A:C|B)_{\rho} \label{eq_tight1} \\
&\leq \inf_{\cR_{B\to BC}}\left \lbrace D\big(\rho_{ABC} \| \cR_{B\to BC}(\rho_{AB}) \big) + \Lambda_{\max}(\rho_{AB}\| \cR_{B\to B}) \right \rbrace \, , \label{eq_ourBound}
\end{align}
where the recovery map $\cP_{B \to BC}$ on the left-hand side is given by~\eqref{eq_explicitRecoveryMap} and the infimum is over all recovery maps $\cR_{B\to BC}$ that map $B$ to $B\otimes C$. The main difference between the lower and upper bound for the conditional mutual information given by~\eqref{eq_tight1} and~\eqref{eq_ourBound}, respectively, is the $\Lambda_{\max}$-term.
\paragraph{Classical case.} 
Inequalities~\eqref{eq_tight1} and~\eqref{eq_ourBound} hold with equality in case $\rho_{ABC}$ is a classical state, i.e.,
\begin{align} \label{eq_classicalState}
\rho_{ABC} = \sum_{a,b,c} P_{ABC}(a,b,c) \proj{a}_A \otimes \proj{b}_B \otimes \proj{c}_C \, ,
\end{align}
for some probability distribution $P_{ABC}$. To see this, we first note that if $\rho_{ABC}$ is classical (in which case $\rho_{ABC}$ and all its marginals commute pairwise) a straightforward calculation gives
\begin{align} \label{eq_classical_Petz}
I(A:C|B)_{\rho}= D\big(\rho_{ABC} \| \cT_{B\to BC}(\rho_{AB}) \big)   \, ,
\end{align}
for the Petz recovery map $\cT_{B\to BC}$ defined in~\eqref{eq_PetzRecoveryMap}. Furthermore, if $\rho_{ABC}$ is classical $\cT_{B \to BC}(\rho_{AB}) = \rho_{BC} \rho_{B}^{-1}  \rho_{AB}$. We further see that $\tr_C \cT_{B \to BC}(\rho_{AB}) = \cT_{B \to B}(\rho_{AB}) =\rho_{AB}$ and hence
\begin{align}
\Lambda_{\max}(\rho_{AB}\| \cT_{B\to B}) =0 \, .
\end{align}
 This shows that in the classical case~\eqref{eq_ourBound} is an equality and that the Petz recovery map $\cT_{B \to BC}$ minimizes the right-hand side of~\eqref{eq_ourBound}. 
 
We further note that in the classical case the measured relative entropy coincides with the relative entropy and the rotated Petz recovery map $\cP_{B\to BC}$ that satisfies~\eqref{eq_tight1} simplifies to the Petz recovery map~$\cT_{B\to BC}$. This together with~\eqref{eq_classical_Petz} then shows that~\eqref{eq_tight1} holds with equality in the classical case.  
\paragraph{Necessity of the $\Lambda_{\max}$-term.}
A natural question regarding~\eqref{eq_MainResIntro} is whether the $\Lambda_{\max}$-term is necessary. Here we show that this is indeed the case by constructing an example proving that a large conditional mutual information does not imply that all recovery maps are bad and hence the $\Lambda_{\max}$-term is indispensable.

More precisely, in Section~\ref{sec_constructionExample} we construct a generic example showing that for any constant $\kappa <\infty$ there exists a classical state $\rho_{ABC}$ (i.e., a state of the form~\eqref{eq_classicalState}) such that
\begin{align} \label{eq_counterExample1}
\kappa \, D_{\max}\big(\rho_{ABC} \| \cR_{B \to BC}(\rho_{AB})\big) < I(A:C|B)_{\rho}   \, ,
\end{align}
for some recovery map $\cR_{B\to BC}$ that satisfies $\cR_{B \to BC}(\rho_{B})=\rho_{BC}$. A similar construction (also given in Section~\ref{sec_constructionExample}) shows that there exists another classical state $\rho_{ABC}$ such that 
\begin{align}  \label{eq_counterExample2}
\kappa \, D_{\max}\big(  \cR_{B \to BC}(\rho_{AB}) \|  \rho_{ABC} \big)  < I(A:C|B)_{\rho}  \, ,
\end{align}
for some recovery map $\cR_{B\to BC}$ that satisfies $\cR_{B \to BC}(\rho_{B})=\rho_{BC}$. 

These constructions therefore show that an additional term like $\Lambda_{\max}(\rho_{AB}\| \cR_{B\to B})$, which measures the deviation from a ``read only'' map on $B$, is necessary to obtain a lower bound on the relative entropy between a state and its reconstructed version. 
The example has an even stronger implication. It shows that the $\Lambda_{\max}$-term is necessary even if one tries to bound the max-relative entropy between a state and its reconstruction version, i.e., $D_{\max}(\rho_{ABC}\| \cR_{B\to BC}(\rho_{AB}))$, which cannot be smaller than $D(\rho_{ABC}\| \cR_{B\to BC}(\rho_{AB}))$, from below.\footnote{The max-relative entropy and its properties are discussed in more detail in Section~\ref{sec_oneShotEntropies}. It is the largest sensible relative entropy measure.}  The two strict inequalities~\eqref{eq_counterExample1} and~\eqref{eq_counterExample2} show that the $\Lambda_{\max}$-term is also necessary if one would allow for swapping the two arguments of the relative (or even max-relative) entropy. Furthermore, restricting the set of recovery maps such that they satisfy $\cR_{B\to BC}(\rho_{B})=\rho_{BC}$ still requires the $\Lambda_{\max}$-term. 

Since for classical states~\eqref{eq_classical_Petz} holds, these examples also show that for the task of minimizing the relative entropy between $\rho_{ABC}$ and its reconstructed state $\cR_{B\to BC}(\rho_{AB})$ the Petz recovery map can be far from being optimal --- even in the classical case. The examples further show that considering recovery maps that leave the $B$ system invariant (i.e., they only ``read'' the $B$-part) is a considerable restriction.
 We refer to Section~\ref{sec_constructionExample} for more information about these examples.

\paragraph{Optimality of the $\Lambda_{\max}$-term.} 
Even in the case where $\rho_{ABC}$ is not classical,~\eqref{eq_MainResIntro} is still close to optimal. We present two arguments why this is the case. First, we show that the $\Lambda_{\max}$-term cannot be replaced by a relative entropy measure that is smaller than the max-relative entropy. More precisely,~\eqref{eq_MainResIntro} is violated if the max-relative entropy in the definition of $\Lambda_{\max}(\rho_{AB} \| \cR_{B\to B})$ is replaced with any $\alpha$-R\'enyi relative entropy for any $\alpha \in [\frac{1}{2},\infty)$. We refer to Section~\ref{sec_Ex_tight} for more information. 

The $\Lambda_{\max}$-term in~\eqref{eq_MainResIntro} quantifies the max-relative entropy distance between $\rho_{AB}$ and its closest state that is invariant under $\cR_{B\to B}$. A natural question is if~\eqref{eq_MainResIntro} remains valid if the $\Lambda_{\max}$-term is replaced by the (max-relative entropy) distance between $\rho_{A B}$ and $\cR_{B \to B}(\rho_{A B})$, i.e., $  D_{\max}(\rho_{A B} \| \cR_{B \to B}(\rho_{A B}))$. This however is ruled out. To see this we recall that by the example mentioned above in~\eqref{eq_counterExample1} there exists a tripartite state $\rho_{ABC}$ and a recovery map $\cR_{B \to BC}$ such that
\begin{align}
2 D_{\max}\bigl(\rho_{A B C} \| \cR_{B \to B C}(\rho_{A B})\bigr) < I(A: C | B)_{\rho}  \, .
\end{align}
The data-processing inequality for the max-relative entropy~\cite{datta09,marco_book} and the fact that the max-relative entropy cannot be smaller than the relative entropy then imply
\begin{align}
   D\bigl(\rho_{A B C} \| \cR_{B \to B C}(\rho_{A B})\bigr) < I(A: C | B)_{\rho} - D_{\max}\bigl(\rho_{A B} \| \cR_{B \to B}(\rho_{A B}) \bigr) \, ,
\end{align}
which shows that~\eqref{eq_MainResIntro} is no longer valid for the modified $\Lambda_{\max}$-term described above.

\subsection{Related results} 
Using the continuity of the conditional entropy, it is possible to derive an upper bound for the conditional mutual information of a state $\rho_{ABC}$ in terms of its distance to any reconstructed state $\sigma_{ABC}:=\cR_{B\to BC}(\rho_{AB})$, where $\cR_{B \to BC}$ denotes an arbitrary recovery map~\cite{berta15,FR14}. This leads to a lower bound on the relative entropy between $\rho_{ABC}$ and $\cR_{B\to BC}(\rho_{AB})$ that however depends on the dimension of the $A$ system. To see this, let $[0,1] \ni x \mapsto h(x):=-x \log x -(1-x) \log (1-x)$ denote the \emph{binary entropy function} and let $\Delta(\tau,\omega):=\frac{1}{2} \norm{\tau-\omega}_1$ be the \emph{trace distance} between $\tau$ and $\omega$. 
The data-processing inequality~\cite{LieRus73_1,LieRus73} implies that
\begin{align}
I(A:C|B)_{\rho} = H(A|B)_{\rho} -H(A|BC)_{\rho} \leq H(A|BC)_{\sigma} - H(A|BC)_{\rho} \, .
\end{align}
By the improved Alicki-Fannes inequality~\cite[Lemma~2]{Winter2016} we find
\begin{align}
I(A:C|B)_{\rho} 
&\leq 2 \Delta(\rho,\sigma) \log(\dim A) + \big(1+ \Delta(\rho,\sigma)\big) h\!\left( \frac{\Delta(\rho,\sigma)}{1+\Delta(\rho,\sigma)} \right)\\
&\leq 2 \sqrt{\Delta(\rho,\sigma)} \big( \log(\dim A) + 1 \big) \, ,
\end{align}
where we used that $(1+x) h(\frac{x}{1+x}) \leq 2 \sqrt{x}$ for all $x \in [0,1]$ and $\Delta(\rho,\sigma) \in [0,1]$. Together with Pinsker's inequality~\cite{Pinsker60,Csiszar67} this gives
\begin{align} \label{eq_dimensionUB}
D\big(\rho_{ABC} \| \cR_{B \to BC}(\rho_{AB}) \big)
\geq \frac{2}{\ln 2} \Delta\big(\rho_{ABC},\cR_{B\to BC}(\rho_{AB})\big)^2 
\geq \frac{I(A:C|B)^4_{\rho}}{8 \ln 2 \big(\log(\dim A) +1 \big)^4}\, .
\end{align}
The fact that this bound explicitly depends on the dimension of the system $A$ is unsatisfactory. Furthermore, the example discussed above in~\eqref{eq_counterExample1} shows that such a dependence on the dimension is unavoidable.

A different approach to derive an upper bound for the conditional mutual information of a state $\rho_{ABC}$ in terms of its distance to a reconstructed state $\cR_{B\to BC}(\rho_{AB})$ was taken in~\cite[Theorem~11 and Remark~12]{DW15} (see also~\cite[Proposition~F.1]{SBT16}). It was shown that for any state $\rho_{ABC}$
\begin{align} \label{eq_badUB}
\int_{-\infty}^{\infty} \beta_0(\di t) \, \bar D_2 \big(\rho_{ABC} \| \cP_{B \to BC}^{[t]}(\rho_{AB}) \big) \geq I(A:C|B)_{\rho} \, ,
\end{align}
where $\beta_0$ and $\cP_{B \to BC}^{[t]}$ are given in~\eqref{eq_beta0} and~\eqref{eq_explicitRecoveryMap}, respectively and $\bar D_2(\tau \| \omega) := \log \tr \, \tau^2 \omega^{-1}$ denotes Petz' R\'enyi relative entropy of order $2$~\cite{Petz86}. 
The examples discussed above imply that the left-hand side of~\eqref{eq_badUB} can be much larger than the relative entropy between $\rho_{ABC}$ and $\cR_{B\to BC}(\rho_{AB})$ for the optimal recovery map $\cR_{B \to BC}$. In other words, rotated Petz recovery maps are generally far from optimal recovery maps.
\section{One-shot relative entropies} \label{sec_oneShotEntropies}
The goal of this section is to derive a triangle-like inequality for the relative entropy (see Lemma~\ref{lem_triangleD}) which will be used in the proof of our main result, i.e., Theorem~\ref{thm_main}. To understand Lemma~\ref{lem_triangleD} we need to review a few properties of one-shot relative entropy measures.

\subsection{Preliminaries}
Let $\St(A)$ and $\Pos(A)$ denote the set of density and nonnegative operators on $A$, respectively. For any linear operator $L$ on $A$, the \emph{trace norm} is given by $\norm{L}_1 := \tr |L|$ with $|L|:=\sqrt{L^\dagger L}$. For $\rho,\sigma \in \Pos(A)$ we write $\rho \ll \sigma$ if the support of $\rho$ is contained in the support of $\sigma$. Within this document our Hilbert spaces are assumed to be separable.
We define the \emph{min-relative entropy}~\cite{renner_phd} as
\begin{align}
D_{\min}(\rho \| \sigma):=-\log \norm{\sqrt{\rho} \sqrt{\sigma}}^2_1  = - \log F(\rho,\sigma)
\end{align}
and the \emph{max-relative entropy}~\cite{datta09,renner_phd} as
\begin{align}
D_{\max}(\rho \| \sigma):=\inf \{ \lambda \in \R : \rho \leq 2^{\lambda} \sigma \} \, .
\end{align}
As the names suggest, the min-relative entropy cannot be larger than the max-relative entropy, or more precisely we have
\begin{align} \label{eq_DminMax}
D_{\min}(\rho\|\sigma) \leq D(\rho \| \sigma) \leq D_{\max}(\rho\| \sigma) \, ,
\end{align}
with strict inequalities in the generic case~\cite{MLDSFT13,marco_book}. The max-relative entropy turns out to be the largest relative entropy measure that satisfies the data-processing inequality and is additive under tensor products~\cite[Section~4.2.4]{marco_book}. We also note that it follows immediately from the definition that  the max-relative entropy cannot increase if the same positive map is applied to both arguments (see also~\cite[Theorem~2]{Hermes2017}  for a more general statement).

The min- and max-relative entropy can be seen as the extreme points of a family of relative entropies called \emph{minimal quantum R\'enyi relative entropy} (also known as \emph{sandwiched R\'enyi  relative entropy})~\cite{MLDSFT13,wilde_strong_2014}. For $\alpha \in [\frac{1}{2},1) \cup (1,\infty)$ and $\rho,\sigma \in \Pos(A)$, this family is defined as
\begin{align}  \label{eq:min_entropy} 
D_{\alpha}(\rho \| \sigma):= \begin{cases}
\frac{1}{\alpha-1} \log  \frac{1}{\tr \rho} \tr \Big(\sigma^{\frac{1-\alpha}{2 \alpha}} \rho \sigma^{\frac{1-\alpha}{2 \alpha}} \Big)^\alpha &\text{if $ \rho \ll \sigma \lor \alpha<1 $}\\
\infty &\text{otherwise} \, .
\end{cases}
\end{align}
It can be shown~\cite{MLDSFT13} that
\begin{align}
D_{\frac{1}{2}}(\rho \| \sigma) = D_{\min}(\rho \| \sigma), \quad \lim_{\alpha \to 1} D_{\alpha}(\rho \| \sigma) = D(\rho \| \sigma), \quad \text{and} \quad \lim_{\alpha \to \infty} D_{\alpha}(\rho \| \sigma) = D_{\max}(\rho \| \sigma) \, .
\end{align}
Furthermore the minimal quantum R\'enyi relative entropy is monotone in $\alpha \in [\frac{1}{2},\infty)$~\cite[Theorem~7]{MLDSFT13}, i.e.,
\begin{align} \label{eq_mono}
D_{\alpha}(\rho \| \sigma) \leq D_{\alpha'}(\rho \| \sigma) \quad \text{for} \quad \alpha \leq \alpha' \, .
\end{align}

\subsection{Triangle-like inequality for relative entropy}
It is well-known that the relative entropy does not satisfy the triangle inequality. For the three (classical) qubit states $\rho=\frac{1}{2} \proj{0} + \frac{1}{4} \id_2$, $\sigma=\frac{1}{2} \proj{1} + \frac{1}{4} \id_2$, and $\omega = \frac{ 1}{2} \id_2$ we have $D(\rho \| \sigma) > D(\rho \| \omega) + D(\omega \| \sigma)$. The following lemma proves a triangle-like inequality for the minimal quantum R\'enyi relative entropy.
\begin{lemma} \label{lem_triangleD}
Let $A$ be a separable Hilbert space, let $\rho \in \St(A)$, $\sigma,\omega \in \Pos(A)$ and let $\alpha \in [\frac{1}{2},1]$. Then
\begin{align} \label{eq_triangleD}
D_{\alpha}(\rho \| \sigma) \leq D_{\alpha}(\rho \| \omega) + D_{\max}(\omega \| \sigma) \, .
\end{align}
\end{lemma}
\begin{proof}
For $\alpha \in [\frac{1}{2},1)$, the function $t \mapsto t^{\frac{1-\alpha}{\alpha}}$ is operator monotone on $[0,\infty)$~\cite[Theorem~V.1.9]{bhatia_book}. Furthermore, the function $X \mapsto \tr X^{\alpha}$ is monotone on the set of Hermitian operators on a separable Hilbert space, since the function $X\mapsto X^\alpha$ is operator monotone~\cite{bhatia_book}. By definition of the max-relative entropy we find
\begin{align}
D_{\alpha}(\rho \| \sigma) 
= \frac{1}{\alpha-1} \log \tr \Big( \rho^{\frac{1}{2}} \sigma^{\frac{1-\alpha}{\alpha}} \rho^{\frac{1}{2}} \Big)^{\alpha} 
\leq D_{\alpha}(\rho \|\omega) + D_{\max}(\omega \| \sigma) \, . 
\end{align}
for $\alpha < 1$. The case $\alpha=1$ then follows by continuity.
\end{proof}
\begin{remark}
We note that if $A$ is a finite-dimensional Hilbert space then~\eqref{eq_triangleD}  is valid for all $\alpha   \in [\frac{1}{2},\infty)$. This follows from the fact that $t \mapsto t^{\frac{1-\alpha}{\alpha}}$ is operator anti-monotone~\cite{marco_book} for $\alpha > 1$ and that the function $X \mapsto \tr X^{\alpha}$ is monotone on the set of Hermitian operators~\cite[Theorem~2.10]{carlen_book}.
\end{remark}

Very recently, a similar triangle-like inequality for R\'enyi relative entropies that additionally involves trace-preserving completely positive maps has been established in~\cite{Christandl2017}. 
The following remarks show that Lemma~\ref{lem_triangleD} is optimal and that there is not much flexibility to prove triangle-like inequalities for the relative entropy different than~\eqref{eq_triangleD}.
\begin{remark} \label{rmk_triangle}
Lemma~\ref{lem_triangleD} is optimal in the sense that~\eqref{eq_triangleD} is no longer valid if $D_{\max}$ is replaced with $D_{\alpha}$ for any $\alpha  \in [\frac{1}{2}, 1]$. To see this, let $p \in (0,1)$ and consider  three classical distributions on $\{0,1\} \times \{0,1\}$ defined by
\begin{align}
P_{XY}(x,y) := \left \lbrace \begin{array}{l l}
p & \text{if } x=y=0 \\
\frac{1-p}{3} & \text{otherwise}
\end{array} \right., \!
  \quad Q_{XY}(x,y):= \left \lbrace \begin{array}{l l}
p & \text{if } x=y=1 \\
\frac{1-p}{3} & \text{otherwise}
\end{array} \right. , \! \quad S_{XY}(x,y)=\frac{1}{4}  \, .
\end{align}
A simple calculation shows that 
\begin{align}
&D(P \| Q) = \frac{4p-1}{3} \log \frac{3p}{1-p} \\
&D(P \| S) = p \log 4p + (1-p) \log \frac{4 (1-p)}{3} \\
&D_{\alpha}(S \| Q)= \frac{1}{\alpha -1} \log \left( \frac{3^{\alpha }}{4^{\alpha} (1-p)^{\alpha -1}} + \frac{1}{4^{\alpha} p^{\alpha-1}} \right) \, .
\end{align}
Choosing $p=1-2^{-\alpha}$ reveals that 
\begin{align} \label{eq_strictt}
D(P \| Q) > D(P \| S) + D_{\alpha}(S \| Q) \quad \text{for all} \quad  \alpha \in [\tfrac{1}{2},\infty) \, .
\end{align}
In the limit $\alpha \to \infty$ the strict inequality \eqref{eq_strictt} becomes an equality.
\end{remark}

\begin{remark}
The statement of Lemma~\ref{lem_triangleD} is no longer true if the max-relative entropy and the relative entropy on the right-hand side of~\eqref{eq_triangleD} are exchanged.
To see this consider the three classical binary probability distributions 
\begin{align}
P(x) := \left \lbrace \begin{array}{l l} 1-p & \text{if } x=0\\
p & \text{otherwise}
\end{array} \right.\!\! , \! \quad Q(x) := \left \lbrace \begin{array}{l l} 1-\eps & \text{if } x=0\\
\eps & \text{otherwise}
\end{array} \right. \!\! , \! \quad S(x) := \left \lbrace \begin{array}{l l} 1-\frac{p}{2} & \text{if } x=0\\
\frac{p}{2} & \text{otherwise}
\end{array} \right. \!, 
\end{align}
with $p,\eps \in (0,1)$. This gives
\begin{align}
&D(P \| Q) =  (1-p) \log \frac{1-p}{1-\eps} + p \log \frac{p}{\eps}\\
&D(S \| Q)= \frac{2-p}{2} \log \frac{2-p}{2-2\eps} + \frac{p}{2} \log \frac{p}{2 \eps} \\
&D_{\max}(P \| S) = \max \left \lbrace\log \frac{2(1-p)}{2-p}, \log 2 \right \rbrace=1 \, .
\end{align}
For $p=\frac{7}{8}$ and $\eps = \frac{1}{8}$ we find that
\begin{align}
D(P\| Q) > D_{\max}(P \| S) + D(S \| Q) \, .
\end{align}
This shows that it is crucial which term in Lemma~\ref{lem_triangleD} carries a max-relative entropy.
\end{remark}

\begin{remark}
The relative entropy satisfies a triangle-like inequality different from Lemma~\ref{lem_triangleD}. For the \emph{log-Euclidean $\alpha$-R\'enyi divergence} $D^\flat_{\alpha}(\omega \| \sigma):=\frac{1}{1-\alpha} \log \tr \, \ee^{\alpha \log \rho +(1-\alpha)\log \sigma}$ it is known~\cite{milan14} that 
\begin{align}
D(\rho \| \sigma) \leq \frac{\alpha}{\alpha -1} D(\rho \| \omega) + D^\flat_{\alpha}(\omega \| \sigma) \quad \text{for} \quad \alpha \in (1,\infty) \, .
\end{align}
We also note that $D_{\infty}^\flat(\omega \| \sigma) \leq D_{\max}(\omega \| \sigma)$ which shows that in the limit $\alpha \to \infty$ we obtain Lemma~\ref{lem_triangleD} for the case $\alpha =1$.
\end{remark}

\section{Main result and proof}
\begin{theorem} \label{thm_main}
Let $A$, $B$, and $C$ be separable Hilbert spaces, let $\rho_{ABC}\in\St(A\otimes B \otimes C)$, and let $\cR_{B \to BC}$ be a trace-preserving completely positive map from $B$ to $B\otimes C$. Then
\begin{align} \label{eq_main}
D\big(\rho_{ABC} \| \cR_{B\to BC}(\rho_{AB}) \big)+  \Lambda_{\max}(\rho_{AB}\| \cR_{B \to B})  \geq I(A:C|B)_{\rho} \, .
\end{align}
\end{theorem}

The quantity $\Lambda_{\max}(\rho_{AB}\| \cR_{B \to B})$ is defined in~\eqref{eq_Lambda} and $\cR_{B \to B} := \tr_C \circ \cR_{B \to BC}$. To prove the assertion of Theorem~\ref{thm_main} we make use of a known lemma stating that the conditional mutual information of a tripartite density operator is bounded from above by the smallest relative entropy distance to Markov chains. Let $\QMC(A\otimes B \otimes C)$ denote the set of Markov chains on $A\otimes B \otimes C$, i.e., tripartite density operators $\rho_{ABC} \in \St(A\otimes B \otimes C)$ that satisfy~\eqref{eq_Markov}.
\begin{lemma}[{\cite[Theorem~4]{ILW08}}] \label{lem_winter}
Let $\rho_{ABC} \in \St(A\otimes B \otimes C)$. Then
\begin{align} \label{eq_winter}
I(A:C|B)_{\rho} \leq \inf_{\mu  \in \QMC}D(\rho_{ABC} \| \mu_{ABC}) \, .
\end{align}
\end{lemma}
\begin{proof}
The proof we provide here follows the lines of a proof by Jen\u{c}ov\'a  (see the short note after the acknowledgements in~\cite{ILW08}), but extends it to general separable spaces. 

Let $\mu_{A B C} \in \QMC$ and assume without loss of generality that the relative entropy $D(\rho_{A B C} \| \mu_{A B C})$ is finite. (If there is no such state then the infimum in~\eqref{eq_winter} equals infinity and the statement is trivial.) Due to the data processing inequality~\cite{lindblad75,uhlmann77} we have
\begin{align}
  0 \leq D(\rho_{B} \| \mu_{B}) \leq D(\rho_{AB} \| \mu_{AB}) \leq D(\rho_{A B C} \| \mu_{A B C}) \label{eq_dataprocessingp} 
\end{align}
and 
\begin{align}
  0 \leq D(\rho_{B} \| \mu_{B}) \leq D(\rho_{BC} \| \mu_{BC}) \leq D(\rho_{A B C} \| \mu_{A B C}) \label{eq_dataprocessingpp} \ .
\end{align}
In particular, the relative entropies $D(\rho_{AB} \| \mu_{AB})$, $D(\rho_{BC} \| \mu_{BC})$, and $D(\rho_{B} \| \mu_{B})$ are finite.
We thus have
\begin{multline*}
  D(\rho_{ABC} \| \mu_{ABC}) + D(\rho_{B} \| \mu_{B}) - D(\rho_{AB} \| \mu_{AB}) - D(\rho_{BC} \| \mu_{BC}) \\
  = \tr\Bigl( \rho_{A B C} \bigl( \log \rho_{A B C} - \log \mu_{A B C} + \log \rho_B - \log \mu_B - \log \rho_{A B} + \log \mu_{A B} - \log \rho_{B C} + \log \mu_{B C} \bigr) \Bigr) \ .
\end{multline*}
Using the Markov chain property~\eqref{eq_MarkovDec} for $\mu_{A B C}$, i.e., 
\begin{align}
\mu_{ABC} = \bigoplus_j P(j) \mu_{Ab_j^L} \otimes \mu_{b_j^R C} \quad \text{for} \quad B  = \bigoplus_{j} b_j^L \otimes b_j^R \, ,
\end{align}
it is straightforward to verify that
\begin{align}
   \log \mu_{ABC} + \log \mu_B -  \log \mu_{AB} - \log \mu_{BC}  = 0 \ .
\end{align}
The above can thus be simplified to
\begin{multline*}
  D(\rho_{ABC} \| \mu_{ABC}) + D(\rho_{B} \| \mu_{B}) - D(\rho_{AB} \| \mu_{AB}) - D(\rho_{BC} \| \mu_{BC}) \\
   = \tr\Bigl( \rho_{A B C} \bigl( \log \rho_{A B C}  + \log \rho_B  - \log \rho_{A B}  - \log \rho_{B C}  \bigr) \Bigr) = I(A : C | B) \ .
\end{multline*}
It follows from~\eqref{eq_dataprocessingp} and~\eqref{eq_dataprocessingpp} that
\begin{align}
    D(\rho_{ABC} \| \mu_{ABC}) 
   \geq  I(A : C | B) \ ,
\end{align}
which concludes the proof. 
\end{proof}

In order to prove Theorem~\ref{thm_main} we need one more lemma that relates the distance to Markov chains and the $\Lambda_{\max}$-quantity defined in~\eqref{eq_Lambda}.
 \begin{lemma} \label{lem_conn}
   Let $\rho_{AB} \in \Pos(A\otimes B)$ and $\cR_{B \to B C}$ be a trace-preserving completely positive map. Then
\begin{align} \label{eq_conn}
     \inf_{\mu \in \QMC} D_{\max}\bigl( \cR_{B \to B C}(\rho_{A B}) \| \mu_{A B C} \bigr) \leq \Lambda_{\max}( \rho_{A B} \|  \cR_{B \to B}) \, .
\end{align}
 \end{lemma}
\begin{proof}
For the proof, we first assume that the system $A$ has a finite dimension, so that conditional entropies of the form $H(A|B)$ are finite. The data processing inequality for the max-relative entropy~\cite{datta09,frankLieb13,marco_book} implies that
  \begin{align}
    &\inf_{\mu_{ABC}} \{ D_{\max}\bigl( \cR_{B \to BC}(\rho_{A B}) \| \mu_{A B C} \bigr) : \,  \mu_{ABC}\in \QMC \}  \nonumber \\
    &\hspace{3mm} \leq   \inf_{\tau_{AB}}\{ D_{\max}\bigl( \cR_{B \to BC}(\rho_{A B}) \| \cR_{B \to B C}(\tau_{A B}) \bigr): \cR_{B\to BC}(\tau_{AB}) \in \QMC, \tau_{AB} \in \St(A\otimes B) \}  \\
    &\hspace{3mm} \leq \inf_{\tau_{AB}} \{ D_{\max}(\rho_{A B} \|\tau_{A B} ) : \cR_{B\to BC}(\tau_{AB}) \in \QMC, \tau_{AB} \in \St(A\otimes B) \}  \, . \label{eq_Dmaxrhomu}
  \end{align}
  Furthermore, because the data processing inequality for the conditional entropy~\cite{LieRus73_1,LieRus73} implies that $H(A|BC)_{\cR_{B \to BC}(\tau_{A B})} \geq H(A|B)_{\tau_{A B}}$ for any $\tau_{AB} \in \St(A\otimes B)$, we also have
  \begin{align}
    \tau_{A B} \in \mathrm{Inv}(\cR_{B \to B})
    & \quad \implies \quad
    H(A|BC)_{\mu} \geq H(A|B)_{\mu}  \quad \text{for $\mu_{A B C} = \cR_{B \to B C}(\tau_{A B})$} \ .
  \end{align}
  Note that the inequality on the right hand side of the implication must, again by the data processing inequality, be an equality, which means that $I(A : C | B)_{\mu} = 0$ and, hence, that $\mu \in \QMC$. This proves the general implication
  \begin{align}
      \tau_{A B} \in \mathrm{Inv}(\cR_{B \to B})
      \quad \implies \quad
      \cR_{B \to B C}(\tau_{A B}) \in \QMC \, .
  \end{align}
  We now use it to obtain
  \begin{align}
    \Lambda_{\max}( \rho_{A B} \| \cR_{B \to B})
    & = 
     \inf_{\tau_{AB}} \{ D_{\max}(\rho_{A B} \|\tau_{A B} ) :\, \tau_{AB} \in \mathrm{Inv}(\cR_{B \to B}) \} \\
     & \geq 
     \inf_{\tau_{AB}} \{ D_{\max}(\rho_{A B} \| \tau_{A B})  : \, \cR_{B \to B C}(\tau_{AB}) \in \QMC, \tau_{AB} \in \St(A\otimes B) \} \, .
  \end{align}
  Combining this with~\eqref{eq_Dmaxrhomu} completes the proof for the case where the system $A$ is finite-dimensional. 
  
  To extend the claim to general separable Hilbert spaces, consider a sequence of finite-rank projectors $(\Pi^{k}_A)_{k \in \mathbb{N}}$ on $A$ with $\Pi^{k}_A \leq \Pi^{k+1}_A$ for any $k \in \mathbb{N}$ that, for $k \to \infty$, converges to the identity in the weak, and hence also in the strong, operator topology~\cite{FAR11}. It follows from the monotonicity of the max-relative entropy under positive maps and~\eqref{eq_conn} for finite-dimensional~$A$ that  
  \begin{align}
     \inf_{\mu \in \QMC} D_{\max}\bigl( \Pi^k_A \cR_{B \to B C}(\rho_{A B}) \Pi^k_A  \| \Pi^k_A \mu_{A B C} \Pi^k_A \bigr) 
     & \leq      \inf_{\mu \in \QMC} D_{\max}\bigl( \Pi^k_A \cR_{B \to B C}(\rho_{A B}) \Pi^k_A  \|  \mu_{A B C}  \bigr) \\
     & \leq \Lambda_{\max}( \Pi^k_A \rho_{A B} \Pi^k_A \|  \cR_{B \to B} ) \, .
 \end{align}
The right hand side can be bounded for any $k\in \N$ by 
\begin{align}
  \Lambda_{\max}(\Pi^k_A \rho_{A B} \Pi^k_A \| \cR_{B \to B} ) 
  & = \inf_{\tau_{A B} \in \mathrm{Inv}(\cR_{B \to B})} D_{\max}(\Pi^k_A \rho_{A B} \Pi^k_A \| \tau_{A B})  \\
  &\leq \inf_{\tau_{A B} \in \mathrm{Inv}(\cR_{B \to B})} D_{\max}\Big(\Pi^k_A \rho_{A B} \Pi^k_A \| \frac{ \Pi_A^k \tau_{A B} \Pi_A^k}{\tr\,  \Pi_A^k \tau_{A B} \Pi_A^k}\Big)\\
   &=  \inf_{\tau_{A B} \in \mathrm{Inv}(\cR_{B \to B})} \left\{ D_{\max}(\Pi^k_A \rho_{A B} \Pi^k_A \| \Pi^k_A  \tau_{A B} \Pi^k_A  ) + \log \tr \,  \Pi_A^k \tau_{A B} \Pi_A^k  \right\} \\
   & \leq \inf_{\tau_{A B} \in \mathrm{Inv}(\cR_{B \to B})} D_{\max}(\Pi^k_A \rho_{A B} \Pi^k_A \| \Pi^k_A  \tau_{A B} \Pi^k_A  ) \, ,
\end{align}
where the first inequality uses that $ \Pi_A^k \tau_{A B} \Pi_A^k / \tr\,  \Pi_A^k \tau_{A B} \Pi_A^k \in \mathrm{Inv}(\cR_{B \to B})$. The final step follows because $\tau_{AB}$ is a density operator and hence $\tr \,  \Pi_A^k \tau_{A B} \leq 1$ for any projector $\Pi_A^k$ on $A$. Using once again the monotonicity of the max-relative entropy under positive maps we find with the above
\begin{align}
  \Lambda_{\max}(\Pi^k_A \rho_{A B} \Pi^k_A \| \cR_{B \to B} ) 
& \leq  \inf_{\tau_{A B} \in \mathrm{Inv}(\cR_{B \to B})} D_{\max}( \rho_{A B} \| \tau_{A B}) \\
& = \Lambda_{\max}(\rho_{A B} \| \cR_{B \to B}) \ .
\end{align}

To conclude the proof, it thus suffices to establish that
\begin{align} \label{eq_leftcontinuity}
       \inf_{\mu \in \QMC} D_{\max}\bigl( \cR_{B \to B C}(\rho_{A B}) \| \mu_{A B C} \bigr)
  \leq \lambda :=
      \limsup_{k \to \infty} \inf_{\mu \in \QMC} D_{\max}\bigl( \Pi^k_A \cR_{B \to B C}(\rho_{A B}) \Pi^k_A  \| \Pi^k_A \mu_{A B C} \Pi^k_A \bigr) \ .
\end{align}

Because the max-relative entropy cannot increase if the same positive map is applied to both arguments, the max-relative entropy is non-decreasing for increasing~$k$, and the $\limsup$ may therefore be replaced by a $\lim$. Hence, there exists a sequence $(\mu^k)_{k \in \mathbb{N}}$  of density operators in $\QMC$ such that 
\begin{align}
  \lambda = \lim_{k \to \infty} D_{\max}\bigl( \Pi^k_A \cR_{B \to B C}(\rho_{A B}) \Pi^k_A  \| \Pi^k_A \mu^k_{A B C} \Pi^k_A \bigr) \ ,
\end{align}
and we can assume without loss of generality that $\Pi^k_A \mu^k_{A B C} \Pi^k_A = \mu^k_{A B C}$. From here we  proceed analogously to the proof of Lemma~11 in~\cite{FAR11}. In particular, we use that the space, $T(\cH)$, of trace-class operators on $\cH = A \otimes B \otimes C$ (equipped with the trace norm) is isometrically isomorphic to the dual of the space $K(\cH)$ of compact operators on $\cH$ (equipped with the operator norm), with the isomorphism $\tau \mapsto \psi_\tau$ given by $\psi_\tau(\kappa) = \tr\,\kappa \tau$, and that, by the Banach-Alaoglu theorem, the closed unit ball on $T(\cH)$ is therefore compact with respect to the weak$^*$ topology. This implies that there exists a subsequence $(\mu^k)_{k \in \Gamma \subset \mathbb{N}}$ that converges in the weak$^*$ topology to an element $\mu \in T(\cH)$, i.e.,
\begin{align} \label{eq_weakconv}
  \lim_{k \to \infty} \tr \, \kappa \mu^k  = \tr \, \kappa \mu \qquad (k \in \Gamma) 
\end{align}
 for al $\kappa \in K(\cH)$. Because, for any $k \in \mathbb{N}$,  $\mu^k$ is a density operator, $\mu$ is also a density operator. The convergence~\eqref{eq_weakconv} also implies
 \begin{align}
   \lim_{k \to \infty}  \tr \, \kappa \Pi^k_A ( 2^{\lambda} \mu^k_{A B C} - \cR_{B \to B C}(\rho_{A B}) ) \Pi^k_A 
    = \tr\, \kappa (2^{\lambda} \mu_{A B C} - \cR_{B \to B C}(\rho_{A B}) )  \qquad (k \in \Gamma) 
 \end{align}
 for any $\kappa \in K(\cH)$. By the definition of the max-relative entropy, the sequence on the left hand side must converge to a non-negative real for any $\kappa \geq 0$. This implies~\eqref{eq_leftcontinuity} .
\end{proof}

\begin{proof}[Proof of Theorem~\ref{thm_main}]
Let $\mu_{ABC}$ be a Markov chain. Combining Lemma~\ref{lem_winter} with Lemma~\ref{lem_triangleD} applied for $\alpha=1$, $\rho=\rho_{ABC}$, $\sigma =\mu_{ABC}$ and $\omega = \cR_{B \to BC}(\rho_{AB})$ gives
\begin{align} \label{eq_midStep}
I(A:C|B)_{\rho} \leq 
D\big(\rho_{ABC} \| \cR_{B\to BC}(\rho_{AB}) \big)  +  \inf_{\mu \in \QMC}D_{\max}\big(\cR_{B\to BC}(\rho_{AB}) \| \mu_{ABC}  \big) \, .
\end{align}
Lemma~\ref{lem_conn} then proves the assertion of Theorem~\ref{thm_main}.\footnote{We note that~\eqref{eq_midStep} is stronger than~\eqref{eq_main} and therefore may be of independent interest.}
\end{proof}

\section{On the tightness of the main result} \label{sec_example}
In this section we construct examples that show two things. First, there exist classical tripartite states with a large conditional mutual information that, however, can be recovered well. This shows the necessity of the $\Lambda_{\max}$-term in the main bound~\eqref{eq_main} --- even if the relative entropy was replaced by the largest possible relative entropy measure, i.e., the max-relative entropy. Furthermore, the violation of such a bound without the $\Lambda_{\max}$-term can be made arbitrarily large. Second, our example shows that~\eqref{eq_main} is no longer valid if the max-relative entropy in the definition of $\Lambda_{\max}(\rho_{AB}\| \cR_{B \to B})$ is replaced with any $\alpha$-R\'enyi relative entropy for any $\alpha \in [\frac{1}{2},\infty)$.

Both examples will be classical, i.e., we consider tripartite states of the form~\eqref{eq_classicalState}. Such states are special as the corresponding density operators of the states and all its marginals are simultaneously diagonalizable. As a result, we can use the classical notion of a distribution to describe such states.

\subsection{A large conditional mutual information does not imply bad recovery} \label{sec_constructionExample}
Let $\cX = \{1,2,\ldots,2^n\}$ for $n \in \N$, $p,q \in [0,1]$ such that $p+q \leq 1$, and consider  two independent random variables $E_Z$ and $E_Y$ on $\{0,1\}$ and $\{0,1,2\}$, respectively, such that $\PP(E_Z=0)=p+q$, $\PP(E_Y=0)=p$, and $\PP(E_Y=1)=q$. Let $X \sim \cU(\cX)$, where $\cU(\cX)$ denotes the uniform distribution on $\cX$ and define two random variables by
\begin{align}
Z:=\left \lbrace \begin{array}{l l}
X & \text{if} \quad E_Z=0 \\
U_Z & \text{otherwise} 
\end{array}
  \right.
  \qquad \text{and} \qquad
  Y:=\left \lbrace \begin{array}{l l}
X & \text{if} \quad E_Y=0 \\
Z & \text{if} \quad E_Y=1 \\
U_Y & \text{otherwise} \,  ,
\end{array}
  \right.
\end{align}
where $U_Y \sim \cU(\cX)$ and $U_Z \sim \cU(\cX)$ are independent. This defines a tripartite distribution $P_{XYZ}$. A simple calculation reveals that
\begin{align}
H(X|Y E_Y E_Z) &= p H(X|X E_Z) + q H(X|Z E_Z) + (1-p-q) H(X|U_Y E_Z) \\
&= q \big( (p+q) H(X|X) +(1-p-q) H(X|U_Z) \big) + (1-p-q) H(X) \\
&= n (1-p-q) (1+q) \, . \label{eq_ddss}
\end{align}
Similarly we find
\begin{align}
H(X|Y Z E_Y E_Z) &=  q(1-p-q) H(X|U_Z) +(1-p-q)(1-p-q) H(X|U_Y) \\
&= n(1-p-q) (1-p) \, .\label{eq_ddss2}
\end{align}
We thus obtain
\begin{align} \label{eq_exCMI}
I(X:Z|Y)_P &= H(X|Y) - H(X|YZ)  \\
& \geq  H(X|Y E_Y E_Z) - H(X|YZ E_Y E_Z) - I(X:E_Y E_Z|YZ) \\
& \geq n (1-p-q) (p+q) - \log 6 \, . \label{eq_CMI_comp}
\end{align}

We next define a recovery map $\cR_{Y \to Y'Z'}$ that creates a tuple of random variables $(Y',Z')$ out of $Y$. Let the recovery map be such that
\begin{align}
(Y',Z'):= (p^2+q+p q) (Y,Y) + \frac{1}{2}\big(1-p^2-q-p q\big) (Y,U) + \frac{1}{2}\big(1-p^2-q-p q\big) (U',Y) \, ,
\end{align}
where $U, U'$ are independent uniformly distributed on $\cX$. Let 
\begin{align}
Q_{XY'Z'}:=\cR_{Y \to Y'Z'}(P_{XY}) 
\end{align}
denote the distribution that is generated when applying the recovery map (described above) to $P_{XY}$. In the following we will assume that $n$ is sufficiently large.
It can be verified easily that $Q_{Y'Z'}=P_{YZ}$. 
Since $P_{XYZ}$ and $Q_{XY'Z'}$ are classical distributions we have $D_{\max}(P_{XYZ} \| Q_{XY'Z'}) = \max_{x,y,z} \log \frac{P_{XYZ}(x,y,z)}{Q_{XY'Z'}(x,y,z)}$. We note that $\PP(X=Y)=p+pq+q^2$ according to the distribution $P_{XY}$ and hence
\begin{align}
&D_{\max}(P_{XYZ}\| Q_{XY'Z'}) \nonumber \\
&= \max \left \lbrace \log \frac{(p+q)^2}{\PP(X=Y)(p^2+q+pq)}, \log \frac{(1-p-q)q}{\PP(X\ne Y)(p^2+q+pq)}, \log \frac{(p+q)(1-p-q)}{\PP(X=Y)\frac{1}{2}(1-p^2-q-pq)}, \right. \nonumber \\
&\hspace{40.0mm} \left.   \log \frac{(1-p-q)p}{\PP(X=Y)\frac{1}{2}(1-p^2-q-pq)}, \log \frac{(1-p-q)^2}{\PP(X\ne Y)(1-p^2-q-pq)} \right \rbrace \label{eq_Dmax_comp1}
\end{align}
and
\begin{align}
&D_{\max}(Q_{XY'Z'}\| P_{XYZ}) \nonumber \\
&= \max \left \lbrace \log \frac{\PP(X=Y)(p^2+q+pq)}{(p+q)^2}, \log \frac{\PP(X\ne Y)(p^2+q+pq)}{(1-p-q)q}, \log \frac{\PP(X=Y)\frac{1}{2}(1-p^2-q-pq)}{(p+q)(1-p-q)}, \right. \nonumber \\
&\hspace{40.0mm} \left.   \log \frac{\PP(X=Y)\frac{1}{2}(1-p^2-q-pq)}{(1-p-q)p}, \log \frac{\PP(X\ne Y)(1-p^2-q-pq)}{(1-p-q)^2} \right \rbrace \, .\label{eq_Dmax_comp2}
\end{align}

We are now ready to state the conclusion of this example.
For $\kappa<\infty$, $p=\frac{1}{2}$, $q=0$, and $n$ sufficiently large we find by combining~\eqref{eq_CMI_comp} with~\eqref{eq_Dmax_comp1}
\begin{align} \label{eq_resEx1}
\kappa \, D_{\max}\big(P_{XYZ}\| \cR_{Y \to YZ}(P_{XY}) \big) = \kappa <\frac{n}{4} - \log 6 \leq I(X:Z|Y)_P \, .
\end{align}
For $\kappa<\infty$, $p=q=\frac{1}{4}$, and $n$ sufficiently large~\eqref{eq_CMI_comp} and~\eqref{eq_Dmax_comp2} imply 
\begin{align}
\kappa\, D_{\max}\big( \cR_{Y \to YZ}(P_{XY}) \| P_{XYZ} \big)  = \kappa \log \frac{15}{8} < \frac{n}{4} - \log 6 \leq I(X:Z|Y)_P  \, .
\end{align}
This shows that there exist classical tripartite distributions $P_{XYZ}$ with a large conditional mutual information $I(X:Y|Z)_P$ and a recovery map $\cR_{Y \to YZ}$ such that $\cR_{Y \to YZ}(P_{XY})$ is close to $P_{XYZ}$ and $\cR_{Y \to YZ}(P_{Y})=P_{YZ}$. The closeness is measured with respect to the max-relative entropy.

\subsection{Tightness of the $\Lambda_{\max}$-term}\label{sec_Ex_tight}
In this section we construct a classical example showing that our main result, i.e.,~\eqref{eq_main} is essentially tight in the sense that it is no longer valid if the max-relative entropy in the definition of $\Lambda_{\max}(\rho_{AB} \| \cR_{B\to B})$, given in~\eqref{eq_Lambda}, is replaced with an $\alpha$-R\'enyi relative entropy for any $\alpha <\infty$. More precisely, for $\alpha \in [1,\infty]$ we define
\begin{align} \label{eq_Lambda_alpha}
  \Lambda_{\alpha}(\rho \| \cE) := \inf_{\tau \in \mathrm{Inv}(\cE)}  D_{\alpha}(\rho \| \tau) \ .
\end{align}
For $\alpha =\infty$ we have $\Lambda_{\infty}(\rho \| \cE) =  \Lambda_{\max}(\rho \| \cE)$. In this section we show that for all $\alpha <\infty$ there exits a (classical) tripartite state $\rho_{ABC}$ and a recovery map $\cR_{B \to BC}$ that satisfies $\cR_{B\to BC}(\rho_B)=\rho_{BC}$ such that
\begin{align} \label{eq_TSquantum}
D\big(\rho_{ABC} \| \cR_{B\to BC}(\rho_{AB}) \big)  < I(A:C|B)_{\rho} -  \Lambda_{\alpha}(\rho_{AB}\| \cR_{B \to B}) \, .
\end{align}

To see this consider the following classical example (where we switch to the classical notation).
 Let $\cS=\{0,\ldots,2^{n}-1 \}$ and consider a tripartite distribution $Q_{XYZ}$ defined via the random variables $X\sim \cU(\cS)$ and $X=Y=Z$. Let $Q'_{XYZ}$ be the distribution defined via the random variables $X \sim \cU(\cS)$, $Y \sim \cU(\cS)$ where $X$ and $Y$ are independent and $Z= (X+Y)\mod 2^n$. For $p\in [0,1]$ we define a binary random variable $E$ such that $\PP(E=0)=p$. Consider the distribution
\begin{align}
P_{XYZ}=\left \lbrace \begin{array}{l l}
Q_{XYZ} & \text{if } E=0 \\
Q'_{XYZ}& \text{if } E=1  \, .
\end{array} \right .
\end{align}
We next define two recovery maps $\tilde \cR_{Y\to Y' Z'}$ and $\bar \cR_{Y \to Y'Z'}$ that create the tuples $(Y',Z')$ out of $Y$ such that 
\begin{align}
(Y',Z')= (Y,Y) \qquad \text{and} \qquad (Y',Z')=\big(U,(Y-U)\!\!\!\!\! \mod 2^n \big) \, ,
\end{align}
where $U \sim  \cU(\cS)$, respectively. We then define another recovery map as 
\begin{align}
\cR_{Y \to Y' Z'}:=p \tilde \cR_{Y \to Y' Z'} +{(1-p)} \bar \cR_{Y \to Y' Z'} \, .
\end{align}
We note that the recovery map satisfies $\cR_{Y\to Y'Z'}(P_{Y})=P_{YZ}$.
A simple calculation shows that
\begin{align}
H(X|YE)_P = p H(X|Y)_Q + (1-p)H(X|Y)_{Q'} = (1-p) n
\end{align}
and
\begin{align}
H(X|YZE)_P = p H(X|YZ)_Q + (1-p)H(X|YZ)_{Q'} = 0 \, .
\end{align}
We thus find
\begin{align}
I(X:Z|Y)_{P} &= H(X|Y) - H(X|YZ) \\
& \geq H(X|YE) - H(X|YZE) - I(X:E|YZ) \\ 
&\geq (1-p)n -h(p)  \, . \label{eq_CMI1}
\end{align}
The distribution $\cR_{Y\to Y'Z'}(P_{XY})$ generated by applying the recovery map to $P_{XY}$ can be decomposed as
\begin{align}
\cR_{Y\to Y'Z'}(P_{XY})= p \left(p \tilde S_{XYZ} + (1-p) \bar S_{XYZ} \right) + (1-p) \left(p \tilde S'_{XYZ} + (1-p) \bar S'_{XYZ} \right) \, ,
\end{align}
where $\tilde S_{XYZ} =\tilde \cR_{Y \to Y' Z'}(Q_{XY})$, $\bar S_{XYZ}=\bar \cR_{Y \to Y' Z'}(Q_{XY})$, $\tilde S'_{XYZ} =\tilde \cR_{Y \to Y' Z'}(Q'_{XY})$, and $\bar S'_{XYZ}=\bar \cR_{Y \to Y' Z'}(Q'_{XY})$.
The joint convexity of the relative entropy~\cite[Theorem~2.7.2]{cover} then implies
\begin{align}
&D\big(P_{XYZ} \| \cR_{Y\to Y'Z'}(P_{XY}) \big) \nonumber \\
& \hspace{10mm} \leq p D\big( Q_{XYZ} \| p \tilde S_{XYZ} + (1-p) \bar S_{XYZ} \big) + (1-p) D\big(Q'_{XYZ} \|  p \tilde S'_{XYZ} + (1-p) \bar S'_{XYZ} \big)
\end{align}
A simple calculation shows that
\begin{align}
D\big( Q_{XYZ} \| p \tilde S_{XYZ} + (1-p) \bar S_{XYZ} \big)
&= \sum_{x=y=z} Q_{XYZ}(x,y,z)\log \frac{Q_{XYZ}(x,y,z)}{p \tilde S_{XYZ}(x,y,z) + (1-p) \bar S_{XYZ}(x,y,z)} \nonumber \\
&\leq \frac{2^{-n}}{p 2^{-n}} = \log \frac{1}{p}
\end{align}
and
\begin{align}
&D\big( Q'_{XYZ} \| p \tilde S'_{XYZ} + (1-p) \bar S'_{XYZ} \big) \nonumber \\
&\hspace{15mm}=\sum_{x,y,z=x+y \!\!\!\!\!\mod 2^n } \!\!\!\!\!Q'_{XYZ}(x,y,z)\log\frac{Q'_{XYZ}(x,y,z)}{p \tilde S'_{XYZ}(x,y,z) + (1-p) \bar S'_{XYZ}(x,y,z)}  \\
&\hspace{15mm}\leq \frac{2^{-2n}}{p 2^{-2n}} = \log \frac{1}{p} \, .
\end{align}
We thus have
\begin{align} \label{eq_relEnt3}
D\big(P_{XYZ} \| \cR_{Y\to Y'Z'}(P_{XY}) \big)\leq \log\frac{1}{p} \, .
\end{align}
We note that the recovery map $\cR_{Y \to Y'} = \tr_{Z'} \circ \cR_{Y \to Y' Z'}$ leaves the uniform distribution $Q'_{XY}$ invariant, i.e., $\cR_{Y \to Y'} (Q'_{XY})= Q'_{XY}$. As a result we find
\begin{align} \label{eq_relEnt2}
\Lambda_{\alpha}(P_{XY} \| \cR_{Y \to Y'}) 
&\leq D_{\alpha}(P_{XY}\| Q'_{XY}) 
=\frac{1}{\alpha -1} \log \big(2^{-n} (1-p)^{\alpha}(2^n -1) + 2^{-n}(1-p+p 2^{n})^\alpha \big) \, ,
\end{align}
where the final step follows by definition of the $\alpha$-R\'enyi relative entropy and a straightforward calculation.

Recall that we need to prove~\eqref{eq_TSquantum}, which in the classical notation reads as
\begin{align} \label{eq_toshow}
D\big(P_{XYZ} \| \cR_{Y\to Y'Z'}(P_{XY}) \big) + \Lambda_{\alpha}(P_{XY} \| \cR_{Y \to Y'})  <I(X:Z|Y)_P  \, ,
\end{align}
for all $\alpha <\infty$.
As mentioned in~\eqref{eq_mono}, the $\alpha$-R\'enyi relative entropy is monotone in $\alpha$ which shows that it suffices to prove~\eqref{eq_toshow} for all $\alpha  \in  (\alpha_0,\infty)$, where $\alpha_0\geq 0$ can be arbitrarily large. 

Combining~\eqref{eq_relEnt3} and~\eqref{eq_relEnt2} shows that for any $\alpha \in (\alpha_0,\infty)$ where $\alpha_0$ is sufficiently large, $p=\alpha^{-2}$, and $n=\alpha$
\begin{align}
D\big(P_{XYZ} \| \cR_{Y\to Y'Z'}(P_{XY}) \big) + \Lambda_{\alpha}(P_{XY} \| \cR_{Y \to Y'})
&\leq 2 \log \alpha + \frac{1}{\alpha-1} \log \left( 1 +2^{-\alpha} (1+ \alpha^{-2} 2^{\alpha})^\alpha \right)\, ,
\end{align}
where we used that $(1-\alpha^{-2})^\alpha(2^\alpha-1) \leq 2^{\alpha}$ for $\alpha \geq 1$.
Using the simple inequality $\log(1+x)\leq \log x + \frac{2}{x}$ for $x\geq1$ gives
\begin{align}
&D\big(P_{XYZ} \| \cR_{Y\to Y'Z'}(P_{XY}) \big) + \Lambda_{\alpha}(P_{XY} \| \cR_{Y \to Y'})  \nonumber \\
&\hspace{30mm} \leq 2 \log \alpha -\frac{\alpha}{\alpha -1} + \frac{\alpha}{\alpha -1} \log\left(1+\frac{2^\alpha}{\alpha^2} \right) + \frac{2}{\alpha -1} 2^\alpha\left(1+\frac{2^\alpha}{\alpha^2} \right)^{-\alpha} \\
&\hspace{30mm} \leq 2 \log \alpha -\frac{\alpha}{\alpha -1} + \frac{\alpha}{\alpha -1} \log\left(1+\frac{2^\alpha}{\alpha^2} \right) + 2^{-\alpha}  \, ,
\end{align}
where the final step is valid since $\alpha$ is assumed to be sufficiently large. 
Using once more $\log(1+x)\leq \log x + \frac{2}{x}$ for $x\geq1$ gives
\begin{align}
&D\big(P_{XYZ} \| \cR_{Y\to Y'Z'}(P_{XY}) \big) +  \Lambda_{\alpha}(P_{XY} \| \cR_{Y \to Y'}) \nonumber \\
&\hspace{30mm}\leq 2 \log \alpha +\frac{\alpha}{\alpha -1} \left( \alpha -2 \log \alpha -1 + \frac{2\alpha^2}{2^\alpha} \right)+ 2^{-\alpha}   \\
&\hspace{30mm}= \alpha -\frac{2}{\alpha-1} \log \alpha + 2^{-\alpha} \mathrm{poly}(\alpha) \, ,
\end{align}
where $\mathrm{poly}(\alpha)$ denotes an arbitrary polynomial in $\alpha$.
As a result, we obtain for a sufficiently large $\alpha$
\begin{align}
D\big(P_{XYZ} \| \cR_{Y\to Y'Z'}(P_{XY}) \big) + \Lambda_{\alpha}(P_{XY} \| \cR_{Y \to Y'}) 
&< \alpha - \frac{2}{\alpha} \label{eq_strict}\\
&\leq \alpha - \alpha^{-1} - h(\alpha^{-2}) \label{eq_binEnt}\\
&\leq I(X:Z|Y)_P \label{eq_finals} \, .
\end{align}
The two steps~\eqref{eq_strict} and~\eqref{eq_binEnt} are both valid because $\alpha$ is sufficiently large. The final step uses~\eqref{eq_CMI1}.

This example shows that~\eqref{eq_main} is no longer valid if the $\Lambda_{\max}$-term is replaced with a $\Lambda_{\alpha}$-term for any $\alpha \in [\frac{1}{2},\infty)$.\footnote{The example does not work in the limit $\alpha \to \infty$.} Note also that this example implies Remark~\ref{rmk_triangle} on the tightness of the triangle-like inequality for the relative entropy.

\section{Open questions}
In this article, we introduced a new entropic quantity~$\Lambda_{\max}(\rho_{AB}\| \cR_{B\to B})$ hat measures how much the map $\cR_{B \to B}$ disturbs the $B$ system, taking system $A$ as a reference. It would be interesting to better understand this quantity and its properties. For example in case $\rho_{ABC}$ is a state whose marginals are all flat\footnote{A state $\omega$ is called flat if all its eigenvalues are either zero or equal to the same constant.}, is it possible to bound $\Lambda_{\max}(\rho_{AB}\| \cR_{B\to B})$ in terms of $D(\rho_{ABC}\|\cR_{B \to BC}(\rho_{AB}))$ from above? This would considerably simplify our main result~\eqref{eq_main} for this special case, which is of interest, e.g.~in applications to condensed matter physics.

\section*{Acknowledgments} 
We thank Mario Berta and Marco Tomamichel for advice on Lemma~\ref{lem_triangleD}. We further thank Aram W.~Harrow, Franca Nester, Joseph M.~Renes, and Volkher B.~Scholz for inspiring discussions. We acknowledge support by the Swiss National Science Foundation (SNSF) via the National Centre of Competence in Research ``QSIT'' and Project No.~$200020\_165843$. We further acknowledge support by the Air Force Office of Scientific Research (AFOSR) via grant FA9550-16-1-0245.

\appendix

\section{Approximate Markov chains can be far from Markov chains} \label{app_closeMarkov}
As mentioned in the introduction, it is known~\cite{CSW12,ILW08} that there exist tripartite states with a small conditional mutual information whose distance to any Markov chain is nevertheless large. For example, consider a state $\rho_{S_1,\ldots S_d} = \proj{\psi}_{S_1,\ldots S_d}$ on $S_1 \otimes \cdots \otimes S_d$ with $\dim S_k = d >1$ for all $k=1,\ldots,d$, where
\begin{align}
\ket{\psi}_{S_1,\ldots S_d} :=\sqrt{\frac{1}{d!}} \sum_{\pi \in \cS_d} \mathrm{sign}(\pi) \ket{\pi(1)} \otimes \ldots \otimes \ket{\pi(d)}
\end{align}
is the \emph{Slater determinant}, $\cS_d$ denotes the group of permutations of $d$ objects, and $\mathrm{sign}(\pi):=(-1)^L$, where $L$ is the number of transpositions in a decomposition of the permutation $\pi$.
The chain rule and the trivial upper bound for the mutual information show that we have
\begin{align}
I(S_1:S_2 \ldots S_d)_{\rho} =  \sum_{k=2}^d I(S_1:S_k | S_2 \ldots S_{k-1})_{\rho}\leq 2 \log d \, .
\end{align}
Because the mutual information is nonnegative, there exists $k \in \{2,\ldots,d\}$ such that 
\begin{align} \label{eq_CMIsmall}
I(S_1:S_k| S_2 \ldots S_{k-1})_{\rho}\leq \frac{2}{d-1} \log d \, ,
\end{align}
which can be arbitrarily small (as $d$ gets large). By definition, the reduced state $\rho_{S_1 S_k}$ is the \emph{antisymmetric state} on $S_1 \otimes S_k$ that is far from separable~\cite[p.~53]{brandao16}. More precisely, for any separable state $\sigma_{S_1 S_k}$ on $S_1 \otimes S_k$  we have $\Delta(\rho_{S_1 S_k},\sigma_{S_1 S_k}) \geq \frac{1}{2}$, where $\Delta(\tau,\omega):=\frac{1}{2} \norm{\tau-\omega}_1$ denotes the \emph{trace distance} between $\tau$ and $\omega$. 
For any state $\mu_{S_1 \ldots S_k}$ on $S_1 \otimes \cdots \otimes S_k$ that forms a Markov chain in order $S_1 \leftrightarrow S_2 \ldots S_{k-1} \leftrightarrow S_k$, it follows by~\eqref{eq_MarkovDec} that its reduced state $\mu_{S_1 S_k}$ on $S_1 \otimes S_k$ is separable. Using the monotonicity of the trace distance under trace-preserving completely positive maps~\cite[Theorem~9.2]{nielsenChuang_book} we thus find
\begin{align}
\Delta(\rho_{S_1 \cdots S_k},\mu_{S_1 \cdots S_k}) \geq \Delta(\rho_{S_1 S_k},\mu_{S_1 S_k}) \geq \frac{1}{2} \, ,
\end{align}
showing that the state $\rho_{S_1 \cdots S_k}$ despite having a conditional mutual information that is arbitrarily small (see~\eqref{eq_CMIsmall}) is far from any Markov chain.

As discussed in the introduction, states with a small conditional mutual information are called approximate Markov chains (which is justified by~\eqref{eq_FR}). The example in this appendix shows that approximate quantum Markov chains are not necessarily close to quantum Markov chains.

\bibliographystyle{arxiv_no_month}
\bibliography{bibliofile}

\begin{thebibliography}{10}

\bibitem{Alberti1983}
P.~M. Alberti.
\newblock A note on the transition probability over ${C}^*$-algebras.
\newblock {\em Letters in Mathematical Physics}, 7(1):25--32, 1983.
\newblock
  \texttt{\href{http://dx.doi.org/10.1007/BF00398708}{DOI:\,10.1007/BF00398708}}.

\bibitem{araki_82}
H.~Araki and T.~Masuda.
\newblock Positive cones and ${L}_p$-spaces for von {N}eumann algebras.
\newblock {\em Publications of the Research Institute for Mathematical
  Sciences}, 18(2):759--831, 1982.
\newblock
  \texttt{\href{http://dx.doi.org/10.2977/prims/1195183577}{DOI:\,10.2977/prims/1195183577}}.

\bibitem{berta16}
M.~Berta, V.~B. Scholz, and M.~Tomamichel.
\newblock R\'enyi divergences as weighted non-commutative vector valued
  ${L}_p$-spaces, 2016.
\newblock \href{https://arxiv.org/abs/1608.05317v1}{arXiv:1608.05317}.

\bibitem{berta15}
M.~Berta, K.~P. Seshadreesan, and M.~M. Wilde.
\newblock R\'enyi generalizations of the conditional quantum mutual
  information.
\newblock {\em Journal of Mathematical Physics}, 56(2):022205, 2015.
\newblock
  \texttt{\href{http://dx.doi.org/10.1063/1.4908102}{DOI:\,10.1063/1.4908102}}.

\bibitem{TB15}
M.~Berta and M.~Tomamichel.
\newblock The fidelity of recovery is multiplicative.
\newblock {\em IEEE Transactions on Information Theory}, 62(4):1758--1763,
  2016.
\newblock
  \texttt{\href{http://dx.doi.org/10.1109/TIT.2016.2527683}{DOI:\,10.1109/TIT.2016.2527683}}.

\bibitem{bhatia_book}
R.~Bhatia.
\newblock {\em Matrix Analysis}.
\newblock Springer, 1997.
\newblock
  \texttt{\href{http://dx.doi.org/10.1007/978-1-4612-0653-8}{DOI:\,10.1007/978-1-4612-0653-8}}.

\bibitem{brandao16}
F.~G.~S.~L. Brand\~ao, M.~Christandl, A.~W. Harrow, and M.~Walter.
\newblock The mathematics of entanglement, 2016.
\newblock \href{https://arxiv.org/abs/1604.01790}{arXiv:1604.01790}.

\bibitem{BHOS14}
F.~G.~S.~L. Brand\~ao, A.~W. Harrow, J.~Oppenheim, and S.~Strelchuk.
\newblock Quantum conditional mutual information, reconstructed states, and
  state redistribution.
\newblock {\em Physical Review Letters}, 115(5):050501, 2015.
\newblock
  \texttt{\href{http://dx.doi.org/10.1103/PhysRevLett.115.050501}{DOI:\,10.1103/PhysRevLett.115.050501}}.

\bibitem{carlen_book}
E.~Carlen.
\newblock {\em Trace Inequalities and Quantum Entropy: An Introductory Course}.
\newblock Contemporary Mathematics, 2009.
\newblock
  \texttt{\href{http://dx.doi.org/10.1090/conm/529}{DOI:\,10.1090/conm/529}}.

\bibitem{Christandl2017}
M.~Christandl and A.~M{\"u}ller-Hermes.
\newblock Relative entropy bounds on quantum, private and repeater capacities.
\newblock {\em Communications in Mathematical Physics}, 353(2):821--852, 2017.
\newblock
  \texttt{\href{http://dx.doi.org/10.1007/s00220-017-2885-y}{DOI:\,10.1007/s00220-017-2885-y}}.

\bibitem{CSW12}
M.~Christandl, N.~Schuch, and A.~Winter.
\newblock Entanglement of the antisymmetric state.
\newblock {\em Communications in Mathematical Physics}, 311(2):397--422, 2012.
\newblock
  \texttt{\href{http://dx.doi.org/10.1007/s00220-012-1446-7}{DOI:\,10.1007/s00220-012-1446-7}}.

\bibitem{cover}
T.~M. Cover and J.~A. Thomas.
\newblock {\em Elements of Information Theory}.
\newblock Wiley Interscience, 2006.

\bibitem{Csiszar67}
I.~Csisz{\'a}r.
\newblock {Information-type measures of difference of probability distributions
  and indirect observations}.
\newblock {\em Studia Sci. Math. Hungar.}, 2:299--318, 1967.

\bibitem{datta09}
N.~Datta.
\newblock Min- and max-relative entropies and a new entanglement monotone.
\newblock {\em IEEE Transactions on Information Theory}, 55(6):2816--2826,
  2009.
\newblock
  \texttt{\href{http://dx.doi.org/10.1109/TIT.2009.2018325}{DOI:\,10.1109/TIT.2009.2018325}}.

\bibitem{DW15}
F.~Dupuis and M.~M. Wilde.
\newblock Swiveled {R}{\'e}nyi entropies.
\newblock {\em Quantum Information Processing}, 15(3):1309--1345, 2016.
\newblock
  \texttt{\href{http://dx.doi.org/10.1007/s11128-015-1211-x}{DOI:\,10.1007/s11128-015-1211-x}}.

\bibitem{FR14}
O.~Fawzi and R.~Renner.
\newblock Quantum conditional mutual information and approximate {M}arkov
  chains.
\newblock {\em Communications in Mathematical Physics}, 340(2):575--611, 2015.
\newblock
  \texttt{\href{http://dx.doi.org/10.1007/s00220-015-2466-x}{DOI:\,10.1007/s00220-015-2466-x}}.

\bibitem{frankLieb13}
R.~L. Frank and E.~H. Lieb.
\newblock Monotonicity of a relative {R}\'enyi entropy.
\newblock {\em Journal of Mathematical Physics}, 54(12):122201, 2013.
\newblock
  \texttt{\href{http://dx.doi.org/10.1063/1.4838835}{DOI:\,10.1063/1.4838835}}.

\bibitem{FAR11}
F.~Furrer, J.~{\AA}berg, and R.~Renner.
\newblock Min- and max-entropy in infinite dimensions.
\newblock {\em Communications in Mathematical Physics}, 306(1):165--186, 2011.
\newblock
  \texttt{\href{http://dx.doi.org/10.1007/s00220-011-1282-1}{DOI:\,10.1007/s00220-011-1282-1}}.

\bibitem{HJPW04}
P.~Hayden, R.~Jozsa, D.~Petz, and A.~Winter.
\newblock Structure of states which satisfy strong subadditivity of quantum
  entropy with equality.
\newblock {\em Communications in Mathematical Physics}, 246(2):359--374, 2004.
\newblock
  \texttt{\href{http://dx.doi.org/10.1007/s00220-004-1049-z}{DOI:\,10.1007/s00220-004-1049-z}}.

\bibitem{ILW08}
B.~Ibinson, N.~Linden, and A.~Winter.
\newblock Robustness of quantum {M}arkov chains.
\newblock {\em Communications in Mathematical Physics}, 277(2):289--304, 2008.
\newblock
  \texttt{\href{http://dx.doi.org/10.1007/s00220-007-0362-8}{DOI:\,10.1007/s00220-007-0362-8}}.

\bibitem{JRSWW15}
M.~Junge, R.~Renner, D.~Sutter, M.~M. Wilde, and A.~Winter.
\newblock Universal recovery maps and approximate sufficiency of quantum
  relative entropy.
\newblock {\em Annales Henri Poincar{\'e}}, 2018.
\newblock
  \texttt{\href{http://dx.doi.org/10.1007/s00023-018-0716-0}{DOI:\,10.1007/s00023-018-0716-0}}.

\bibitem{LieRus73_1}
E.~H. Lieb and M.~B. Ruskai.
\newblock A fundamental property of quantum-mechanical entropy.
\newblock {\em Physical Review Letters}, 30:434--436, 1973.
\newblock
  \texttt{\href{http://dx.doi.org/10.1103/PhysRevLett.30.434}{DOI:\,10.1103/PhysRevLett.30.434}}.

\bibitem{LieRus73}
E.~H. Lieb and M.~B. Ruskai.
\newblock Proof of the strong subadditivity of quantum-mechanical entropy.
\newblock {\em Journal of Mathematical Physics}, 14(12):1938--1941, 1973.
\newblock
  \texttt{\href{http://dx.doi.org/10.1063/1.1666274}{DOI:\,10.1063/1.1666274}}.

\bibitem{lindblad75}
G.~Lindblad.
\newblock Completely positive maps and entropy inequalities.
\newblock {\em Communications in Mathematical Physics}, 40(2):147--151, 1975.
\newblock
  \texttt{\href{http://dx.doi.org/10.1007/BF01609396}{DOI:\,10.1007/BF01609396}}.

\bibitem{milan14}
M.~Mosonyi and T.~Ogawa.
\newblock Strong converse exponent for classical-quantum channel coding, 2014.
\newblock \href{https://arxiv.org/abs/1409.3562}{arXiv:1409.3562}.

\bibitem{Hermes2017}
A.~M{\"u}ller-Hermes and D.~Reeb.
\newblock Monotonicity of the quantum relative entropy under positive maps.
\newblock {\em Annales Henri Poincar{\'e}}, 18(5):1777--1788, 2017.
\newblock
  \texttt{\href{http://dx.doi.org/10.1007/s00023-017-0550-9}{DOI:\,10.1007/s00023-017-0550-9}}.

\bibitem{MLDSFT13}
M.~M\"{u}ller-Lennert, F.~Dupuis, O.~Szehr, S.~Fehr, and M.~Tomamichel.
\newblock On quantum {R}\'enyi entropies: A new generalization and some
  properties.
\newblock {\em Journal of Mathematical Physics}, 54(12), 2013.
\newblock
  \texttt{\href{http://dx.doi.org/http://dx.doi.org/10.1063/1.4838856}{DOI:\,http://dx.doi.org/10.1063/1.4838856}}.

\bibitem{nielsenChuang_book}
M.~A. Nielsen and I.~L. Chuang.
\newblock {\em Quantum Computation and Quantum Information}.
\newblock Cambridge University Press, 2000.

\bibitem{Petz86}
D.~Petz.
\newblock Quasi-entropies for finite quantum systems.
\newblock {\em Reports on Mathematical Physics}, 23(1):57 -- 65, 1986.
\newblock
  \texttt{\href{http://dx.doi.org/http://dx.doi.org/10.1016/0034-4877(86)90067-4}{DOI:\,http://dx.doi.org/10.1016/0034-4877(86)90067-4}}.

\bibitem{Pet86}
D.~Petz.
\newblock Sufficient subalgebras and the relative entropy of states of a von
  {N}eumann algebra.
\newblock {\em Communications in Mathematical Physics}, 105(1):123--131, 1986.
\newblock
  \texttt{\href{http://dx.doi.org/10.1007/BF01212345}{DOI:\,10.1007/BF01212345}}.

\bibitem{Pet03}
D.~Petz.
\newblock Monotonicity of quantum relative entropy revisited.
\newblock {\em Reviews in Mathematical Physics}, 15(01):79--91, 2003.
\newblock
  \texttt{\href{http://dx.doi.org/10.1142/S0129055X03001576}{DOI:\,10.1142/S0129055X03001576}}.

\bibitem{Pinsker60}
M.~S. Pinsker.
\newblock {\em {Information and Information Stability of Random Variables and
  Processes}}.
\newblock Izv. Akad. Nauk, Moskva, 1960.

\bibitem{renner_phd}
R.~Renner.
\newblock Security of quantum key distribution.
\newblock {\em PhD thesis, ETH Zurich}, 2005.
\newblock available at
  \texttt{arXiv:\href{http://arxiv.org/abs/quant-ph/0512258}{quant-ph/0512258}}.

\bibitem{sutter_phd}
D.~Sutter.
\newblock Approximate quantum {M}arkov chains, 2018.
\newblock PhD thesis, ETH Zurich available at
  \href{https://arxiv.org/abs/1802.05477}{arXiv:1802.05477}.

\bibitem{SBT16}
D.~Sutter, M.~Berta, and M.~Tomamichel.
\newblock Multivariate trace inequalities.
\newblock {\em Communications in Mathematical Physics}, 352(1):37--58, 2017.
\newblock
  \texttt{\href{http://dx.doi.org/10.1007/s00220-016-2778-5}{DOI:\,10.1007/s00220-016-2778-5}}.

\bibitem{SFR15}
D.~Sutter, O.~Fawzi, and R.~Renner.
\newblock Universal recovery map for approximate {M}arkov chains.
\newblock {\em Proceedings of the Royal Society of London A: Mathematical,
  Physical and Engineering Sciences}, 472(2186), 2016.
\newblock
  \texttt{\href{http://dx.doi.org/10.1098/rspa.2015.0623}{DOI:\,10.1098/rspa.2015.0623}}.

\bibitem{STH15}
D.~Sutter, M.~Tomamichel, and A.~W. Harrow.
\newblock Strengthened monotonicity of relative entropy via pinched {P}etz
  recovery map.
\newblock {\em IEEE Transactions on Information Theory}, 62(5):2907--2913,
  2016.
\newblock
  \texttt{\href{http://dx.doi.org/10.1109/TIT.2016.2545680}{DOI:\,10.1109/TIT.2016.2545680}}.

\bibitem{marco_book}
M.~Tomamichel.
\newblock {\em Quantum Information Processing with Finite Resources}, volume~5
  of {\em SpringerBriefs in Mathematical Physics}.
\newblock Springer, 2015.
\newblock
  \texttt{\href{http://dx.doi.org/10.1007/978-3-319-21891-5}{DOI:\,10.1007/978-3-319-21891-5}}.

\bibitem{Uhl76}
A.~Uhlmann.
\newblock The ``transition probability'' in the state space of a {*}-algebra.
\newblock {\em Reports on Mathematical Physics}, 9(2):273 -- 279, 1976.
\newblock
  \texttt{\href{http://dx.doi.org/10.1016/0034-4877(76)90060-4}{DOI:\,10.1016/0034-4877(76)90060-4}}.

\bibitem{uhlmann77}
A.~Uhlmann.
\newblock Relative entropy and the {W}igner-{Y}anase-{D}yson-{L}ieb concavity
  in an interpolation theory.
\newblock {\em Communications in Mathematical Physics}, 54(1):21--32, 1977.
\newblock
  \texttt{\href{http://dx.doi.org/10.1007/BF01609834}{DOI:\,10.1007/BF01609834}}.

\bibitem{wilde15}
M.~M. Wilde.
\newblock Recoverability in quantum information theory.
\newblock {\em Proceedings of the Royal Society of London A: Mathematical,
  Physical and Engineering Sciences}, 471(2182):20150338, 2015.
\newblock
  \texttt{\href{http://dx.doi.org/10.1098/rspa.2015.0338}{DOI:\,10.1098/rspa.2015.0338}}.

\bibitem{wilde_strong_2014}
M.~M. Wilde, A.~Winter, and D.~Yang.
\newblock Strong converse for the classical capacity of entanglement-breaking
  and {{Hadamard}} channels via a sandwiched {{R{\'e}nyi}} relative entropy.
\newblock {\em Communications in Mathematical Physics}, 331(2):593--622, 2014.
\newblock
  \texttt{\href{http://dx.doi.org/10.1007/s00220-014-2122-x}{DOI:\,10.1007/s00220-014-2122-x}}.

\bibitem{Winter2016}
A.~Winter.
\newblock Tight uniform continuity bounds for quantum entropies: Conditional
  entropy, relative entropy distance and energy constraints.
\newblock {\em Communications in Mathematical Physics}, 347(1):291--313, 2016.
\newblock
  \texttt{\href{http://dx.doi.org/10.1007/s00220-016-2609-8}{DOI:\,10.1007/s00220-016-2609-8}}.

\end{thebibliography}
      
\end{document}